\documentclass[11pt,draftcls]{IEEEtran}

\usepackage[mathcal]{euscript}
\usepackage{bbm,dsfont,mathrsfs,fixmath}
\usepackage{stmaryrd}
\usepackage{amsmath,amstext,amsfonts,amssymb}
\usepackage{epsfig,float}
\usepackage{graphicx}
\usepackage{cite}
\usepackage{psfrag}
\usepackage{subfigure}
\usepackage{url}
\usepackage{shadow,color,pifont,times,rotate}

\newcommand{\mb}{\mathbf}

\newcommand{\mc}{\mathcal}

\def\ex{\times }

\def\esp{{\mathbb E}}
\newtheorem{theorem}{Theorem}[section]
\newtheorem{lemma}[theorem]{Lemma}

\newtheorem{remark}[theorem]{Remark}
\newtheorem{definition}[theorem]{Definition}
\onecolumn

\begin{document}
 \footernote{To appear in Proc. of IEEE International Symposium on Information Theory (ISIT2010).}
\title{On the Capacity of Compound State-Dependent\\[-2mm] Channels with States Known at the Transmitter}

\author{
\authorblockN{Pablo Piantanida}\\
\authorblockA{Department of Telecommunications, SUPELEC\\
Plateau de Moulon, 91192 Gif-sur-Yvette, France\\
Email: pablo.piantanida@supelec.fr}\\[3mm]
\and
\authorblockN{Shlomo Shamai (Shitz)}\\ 
\authorblockA{Department of Electrical Engineering, Technion - Israel Institute of Technology\\
Technion city, Haifa 32000, Israel\\
Email: sshlomo@ee.technion.ac.il}
}

\maketitle

\begin{abstract}
This paper investigates the capacity of compound state-dependent channels with non-causal state information available at only the transmitter. A new lower bound on the capacity of this class of channels is derived. This bound is shown to be tight for the special case of compound channels with stochastic degraded components, yielding the full characterization of the capacity. Specific results are derived for the compound Gaussian Dirty-Paper (GDP) channel. This model consists of an additive white Gaussian noise (AWGN) channel corrupted by an additive Gaussian interfering signal, known at the transmitter only, where the input and the state signals are affected by fading coefficients whose realizations are unknown at the transmitter.  Our bounds are shown to be tight for specific cases. Applications of these results arise in a variety of wireless scenarios as multicast channels, cognitive radio and  problems with interference cancellation.
\end{abstract}

\section{Introduction}
In the recent years, intensive research addressing theoretical and practical aspects was undertaken on communications over channels controlled by random parameters, namely states. Gel'fand and Pinsker \cite{GF-1980} derived the capacity expression for discrete memoryless channels (DMCs), where the i.i.d. state sequence is known at the transmitter before the start of the communication, but not at the receiver. This scenario is known as state-dependent DMCs with non-causal state information. Costa \cite{costa-1983} considered the case of an additive white Gaussian noise (AWGN) channel corrupted by an additive Gaussian interference which is available at the transmitter only. He showed that choosing an appropriate  probability distribution (PD) for the auxiliary random variable (RV) and the state, referred to as \emph{Dirty-Paper Coding} (DPC), there is no loss in capacity if the interference is known only to the encoder. This result has gained considerable attention because of its potential use to mitigate the interference effects in multi-user scenarios.     

In this work we focus on the compound state-dependent channel with non-causal state information at the transmitter. This channel arises in scenarios where there is uncertainty on the channel statistic. In this model, the conditional PD of the channel is parameterized by $\theta$, which belongs to an arbitrary set $\Theta$ and remains constant during the communication. Whereas, neither the sender nor the receiver  are recognizant of the realization  $\theta$ that governs the communication. This problem was initially investigated in \cite{Mitran-Devroye-Tarokh2006}, where lower and upper bounds on the capacity were derived. In \cite{Steinberg-Shamai-2005}, this problem is identified as being equivalent to the \emph{common-message} broadcast channel (BC) with non-causal state information at the transmitter. Moreover in \cite{Khisti-Lapidoth-Wornell2007}, this is recognized to be the \emph{multicast channel}. Results were obtained for AWGN and binary channels, where a transmitter sends a common message to multiple receivers and each of them experiences an additive interference available at the transmitter only.  These channels are of great interest because of their role in multi-user channels and in particular, for the emerging field of cognitive radios.  Recent work in \cite{Maric-Goldsmith-Kramer-Shamai2008} investigated the capacity of this framework, which is essentially related to the problem considered here when the cognitive user is unaware of the channel path gains. Broadcast channels with imperfect channel knowledge are also instances of this class of channels (cf. \cite{Erez-Shamai-Zamir2005} and \cite{Weingarten-shamai07}).  

In prior work \cite{Piantanida-shamai2009}, \cite{Moulin-Wang07}, it was claimed that a strong converse establishes the optimality of the lower bound first derived in \cite{Mitran-Devroye-Tarokh2006}. In this paper we will demonstrate that this is not the case in general.  In fact, the rate expression \eqref{capacity-low} that was conjectured to be optimal for the general compound channel with states corresponds to the natural extension of the capacity expression obtained by Gel'fand and Pinsker's \cite{GF-1980} to the compound setting case.  Here we establish a new lower bound on the capacity of this class of channels that  can outperform the previous lower bound. This  bound is based on a non-conventional approach \cite{Nair-ELGAmal2009}, \cite{ElGamal09} via a broadcasting strategy that allows the encoder to adapt the auxiliary RVs to each of possible channel outcomes (or each of different users in the multicast setting). Finally, we specialize this bound to the compound Gaussian Dirty-Paper (GDP) channel and derive an upper bound which is tight for some compound models. Furthermore, we show that our lower bound is tight for the compound channel with stochastic degraded components. Recent independent efforts deriving similar results are reported in \cite{Chandra-ElGamal10}, where explicit examples demonstrate also that the rate in expression  \eqref{capacity-low} can be surpassed. The organization of this paper is as follows. Definitions and main results are stated in Section \ref{sec-definitions-results}, while the proof outline and an application to the compound GDP channel are given in Sections \ref{sec-proof} and  \ref{sec-example}.

\section{Problem Statement and Main Results}\label{sec-definitions-results}

In this section, we  introduce main definitions, formalize the problem and present lower and upper bounds on the capacity. 

\subsection{Definitions and Problem Statement}\label{sec-definitions}

We begin with the description of an arbitrary family of channels with discrete input $x\in \mathscr{X}$, discrete state $s\in \mathscr{S}$ and discrete output $y\in \mathscr{Y}$, which is characterized by a set of conditional probability distributions (PDs) $\mc{W}_\Theta =\big\{ W_\theta:\mathscr{X} \ex \mathscr{S}  \longmapsto \mathscr{Y}\big\}_{\theta\in \Theta}$, indexed by $\theta\in \Theta$, where $\Theta$ is the set  of indexes (assumed to be finite). The transition PD of the $n$-memoryless extension  with inputs $\mb{x}=(x_1,\dots,x_n)$, states $\mb{s}=(s_1,\dots,s_n)$ and outputs $\mb{y}=(y_1,\dots,y_n)$ is given by     
\begin{equation}
W_\theta^n(\mb{y}|\mb{x},\mb{s})=\prod\limits_{i=1}^n  W_\theta (y_i|x_i,s_i).         \label{def-compound}   
\end{equation}
The sequence of states $\mb{s}$ is assumed to be drawn i.i.d. with PD $P_S$. The encoder is assumed to know the sequence of states before the transmission starts, but the decoder does not know it. Whereas, neither the sender nor the receiver  are cognizant of the  realization of $\theta$ that governs the communication. The channel states change from letter to letter following the PD $P_S$, but $\theta\in\Theta$ should not change during the communication. This scenario is known as compound DMCs with non-causal state information at the transmitter. We argue  the capacity is not increased if the decoder is aware of the index $\theta\in\Theta$. 

\begin{definition}[Code] \label{def-code}
A code for this channel consists of two mappings, the encoder mapping $\big\{ \varphi:\mc{M}_n \ex \mathscr{S}^n \longmapsto \mathscr{X}^n  \big\}$ and the decoder mapping $\big\{ \psi:\mathscr{Y}^n \longmapsto \mc{M}_n  \big\}$ for some finite set of integers $\mc{M}_n=\big\{ 1,\dots, M_n \big\}$.  The encoding function $\big\{ \varphi \big\}$ maps the corresponding message $m\in\mc{M}_n $ and the states $\mathscr{S}^n$ into $\mathscr{X}^n$ and the decoding function $\big\{ \psi \big\}$ maps $\mathscr{Y}^n$ into $\mc{M}_n$.  In presence of feeback, where the past of the channel outputs are available at the transmitter, the encoder mappings are given by $\big\{ \varphi_i:\mc{M}_n \ex \mathscr{Y}^{i-1} \ex \mathscr{S}^n \longmapsto \mathscr{X}  \big\}_{i=1}^n$.
An $n$-length block code for simultaneous DMCs $\big\{ W_\theta^n:\mathscr{X}^n \ex \mathscr{S}^n  \longmapsto \mathscr{Y}^n, \theta\in \Theta\big\}_ {n=1}^{\infty}$ consists on a common code $(\varphi,\psi)$ for the set of channels $ \mc{W}_{\Theta}^n=\big\{ W_\theta^n \big\}_{\theta\in \Theta}$. The rate of such code is $n^{-1} \log M_n$ and its error probability associated to the message $m\in \mc{M}_n$ is  defined as  
\begin{equation*}
e_{m}^{(n)}\big( W_{\theta}^n, \varphi,\psi| \mb{s} \big) = W^n_\theta\big( \bigcup_{m^{\prime}\neq m}   \psi^{-1}(m^{\prime}) \big | \varphi(m,\mb{s}), \mb{s}   \big) ,\label{errorprob_def}
\end{equation*}
for $\theta\in\Theta$ and $\mb{s}\in  \mathscr{S}^n$. The maximum of average error probability (over all messages) is defined as  
\begin{equation*}
\bar{e}_{\max}^{(n)}\big( \mc{W}_{\Theta}^n, \varphi,\psi\big)=\max_{m\in\mc{M}}\, \max_{\theta\in \Theta} \,\sum\limits_{\mb{s}\in  \mathscr{S}^n} P^n_S(\mb{s}) \, e_{m}^{(n)}\big(W_{\theta}^n, \varphi,\psi | \mb{s}\big).\label{mean_errorprob_def}
\end{equation*}
An $n$-length block code for the simultaneous DMCs $\mc{W}_{\Theta}^n$ whose maximum of average error probability \eqref{mean_errorprob_def} satisfies $\bar{e}_{\max}^{(n)}\big( \mc{W}_{\Theta}^n, \varphi,\psi\big)\leq \epsilon$ will be called an $(n,\epsilon)$-code.          
\end{definition}

\begin{definition}[Achievable rate and capacity] Given $0< \epsilon,\gamma< 1$, a non-negative number $R$ is an $\epsilon$-achievable rate for the compound channel $\mc{W}_{\Theta}$ if for every sufficiently large $n$ there exist $(n,\epsilon)$-codes of rate 
$n^{-1} \log M_n \geq R-\gamma.$ Then, $R$ is an achievable rate if it is $\epsilon$-achievable for every $0<\epsilon<1$. The supremum of $\epsilon$-achievable rates is called the $\epsilon$-capacity $C_\epsilon$ while the supremum of achievable rates is called the capacity.
\end{definition}

In the remainder of this section we state lower and upper bounds on the capacity of the general compound DMC \eqref{def-compound}. \vspace{-1mm}

\subsection{Lower Bounds on the Capacity}\label{chI-sectionII}

The following achievable rate, first found in \cite{Mitran-Devroye-Tarokh2006}, corresponds to the straightforward extension of the Gel'fand and Pinsker's capacity \cite{GF-1980} to the compound setting case.  

\begin{theorem}\label{theo-capacity-low}
A lower bound on the capacity of the compound DMC $\big\{ W_\theta:\mathscr{X} \ex \mathscr{S}  \longmapsto \mathscr{Y}\big\}_{\theta\in \Theta}$ with states non-causally known only at the transmitter  is given by 
\begin{equation}
R = \sup\limits_{P_{XU|S} \in\mathscr{Q}}\, \min_{\theta\in \Theta } \big\{I(U;Y_\theta)-I(U;S) \big\}, 
\label{capacity-low}
\end{equation}
where $ U \minuso (X,S) \minuso Y_\theta$ for all $\theta\in\Theta$ and the set of admissible input PDs is defined as follows $\mathscr{Q}=\big \{P_{XU|S} \in  \mathscr{P}( \mathscr{X} \ex \mathscr{U}):  P_{XU|S} = P_{X|US} P_{U|S}$,  $\| \mathscr{U} \|\leq \|\mathscr{X} \| \|\mathscr{S} \|+\| \Theta \| \big \}$.
\end{theorem}
Notice that if the encoder is unaware of the states, i.e., $(X,U)$ must be independent of $S$, expression \eqref{capacity-low} reduces to the capacity of standard compound DMCs \cite{Lapidoth-1998}. 

We next state a new achievable rate that improves \eqref{capacity-low}. For sake of clarity, we first consider the case of two components $\Theta=\{1,2\}$ and then we generalize this to an arbitrary set $\Theta$.

\begin{theorem}\label{theo-new-lowerbound1}
A lower bound on the capacity of the compound DMC $\big\{ W_1, W_2:\mathscr{X} \ex \mathscr{S}  \longmapsto \mathscr{Y}\big\}$ with states non-causally known only at the transmitter  is given by 
\begin{align} 
R = \sup\,  \min \big\{ &I(U,V_1;Y_1)-I(U,V_1;S),\nonumber\\ 
& I(U,V_2;Y_2)-I(U,V_2;S),\nonumber\\
\frac{1}{2}\big[& I(U,V_1;Y_1)-I(U,V_1;S) +  \nonumber\\
I(U,V_2;Y_2)-& I(U,V_2;S) - I(V_1;V_2|U,S) \big]\big\}, 
\label{capacity-new-lowerbound1}
\end{align} 
where the supremum  is taken over the set of all joint PDs $P_{XUV_1V_2|S} = P_{X|UV_1V_2S} P_{UV_1V_2|S} $ that satisfy $ (U,V_1,V_2) \minuso (X,S) \minuso (Y_1,Y_2)$ form a Markov chain.
\end{theorem}

\begin{remark} Expression \eqref{capacity-new-lowerbound1} can be reduced to  \eqref{capacity-low} by setting $V_1=V_2=U$. In contrast, there is no possible choice for $(X,U)$ in theorem \ref{theo-capacity-low} yielding the rate \eqref{capacity-new-lowerbound1}. This observation implies that expression \eqref{capacity-low} cannot be optimal for the general compound DMC \eqref{def-compound}. Furthermore, we shall see (section \ref{sec-example}) that for the compound GDP channel the RVs $(V_1,V_2)$ are indeed needed. 
\end{remark}

\begin{theorem}\label{theo-new-lowerbound2}
A lower bound on the capacity of the compound DMC $\big\{ W_\theta:\mathscr{X} \ex \mathscr{S}  \longmapsto \mathscr{Y}\big\}_{\theta\in \Theta}$ with general components $\Theta= \{1, \dots, K\}$ and states non-causally known only at the transmitter  is given by 
\begin{align} 
R = \sup\, & \min_{ \mc{K} \subseteq \Theta }\,  \frac{1}{\| \mc{K} \|}  \Big[ \sum\limits_{k \in  \mc{K} } I(U,V_k;Y_k) -  \| \mc{K} \| I(U;S) \nonumber\\
&+ H\big(\{ V_t\, | \, t \in \mc{K}   \} | U,S\big)  -\sum\limits_{k \in  \mc{K}  } H(V_k | U ) \Big], 
\label{capacity-new-lowerbound2}
\end{align} 
where the supremum  is taken over the set of all joint PDs $P_{XUV_1\dots V_K|S} = P_{X|UV_1\dots V_KS} P_{UV_1\dots V_K|S} $ satisfying $ (U,V_1,\dots, V_K) \minuso   (X,S) \minuso (Y_1,\dots,Y_K)$ form a Markov chain.
\end{theorem}
Observe that the rate \eqref{capacity-new-lowerbound2} reduces to the rate  \eqref{capacity-new-lowerbound1} for the case of $K=2$. The proofs  of these theorems are sketched in Section \ref{sec-proof}. We next state capacity results for some special cases.

\subsection{Capacity of Some Compound Channels}\label{chI-sectionIII}

\begin{definition}[Degraded components]
Let  $\big\{ W_1, W_2:\mathscr{X}$ $ \ex \mathscr{S}  \longmapsto \mathscr{Y}\big\}$ be a compound DMC with components $\Theta=\{1,2\}$. It is said to be a \emph{stochastically degraded} \cite{Steinberg2005} if there exists some stochastic mapping $\big\{ \tilde{W}:\mathscr{Y}  \longmapsto \mathscr{Y} \big\}$ such that 
$
W_2(y_2|x,s)=\sum_{y_1\in \mc{Y}} W_1(y_1|x,s)  \tilde{W}(y_2|y_1), 
$
for all $y_2\in \mathscr{Y} $ and every pair $(x,s)\in \mathscr{X} \times\mathscr{S} $. This shall be denoted by $W_2 \preceq W_1$ (i.e. $W_2$ is a degraded version of the channel $W_1$).   
\end{definition}

\begin{theorem}[degraded components] \label{theo-capacity-degraded}
The capacity of the compound DMC $\big\{ W_\theta:\mathscr{X} \ex \mathscr{S}  \longmapsto \mathscr{Y}\big\}_{\theta\in \Theta}$ with components $\Theta= \{1,\dots, K\}$ where $W_{K}\preceq W_{K-1} \preceq  \dots \preceq  W_1$ and states non-causally known only at the transmitter is given by                                       
\begin{equation}
C_{\Theta} = \sup\limits_{P_{X V_1\dots V_ {K}|S}\in\mathscr{Q}_{\textrm{D}}}  \min_{\theta \in \Theta}\big\{I(V_{\theta};Y_\theta)-I(V_\theta;S) \big\},
\label{eq-capa-degraded}
\end{equation}
where the set of admissible input PDs is defined by 

$\mathscr{Q}_{\textrm{D}}=\big \{P_{X V_1\dots V_ {K}|S} \in  \mathscr{P}(  \mathscr{X}\ex \mathscr{V}_1 \ex \dots  \ex \mathscr{V}_K  ): $ 
$ P_{XV_1\dots V_ {K}|S}  = P_{X|SV_1}P_{V_1|SV_2} \dots P_{V_{K-1}|SV_K} $ \\ $P_{V_K|S},\,(V_1,\dots , V_ {K}) \minuso  (X,S) $ $ \minuso (Y_1,\dots, Y_ {K}) \big \}.$
\end{theorem}

\begin{proof}
For the case $\| \Theta \|=2$, the direct proof follows by choosing $V_2=U$ in theorem \ref{theo-new-lowerbound1} and the converse proof follows by linking together the outputs  $(Y_1,\dots,Y_\theta)$. Whereas this proof procedure  easily extends to an arbitrary set $\Theta$.   \vspace{1mm}
\end{proof}

\begin{theorem}[feedback] \label{theo-capacity-feedback}
The capacity of the compound DMC $\mc{W}_{\Theta}$ with states non-causally known only at the transmitter and feedback is given by 
\begin{equation}
C_{\textrm{FB}} = \min_{\theta\in \Theta } \sup\limits_{P_{XU_{\theta} |S} \in\mathscr{Q}}\,  \big\{I(U_{\theta};Y_\theta)-I(U_\theta;S) \big\}.
\label{eq-capa-feedback}
\end{equation}
\end{theorem}
This theorem easily follows from \cite{GF-1980} and by observing  that the encoder is able to estimate the channel from the feedback.

\section{Sketch of  Proof of theorems \ref{theo-new-lowerbound1}  and \ref{theo-new-lowerbound2} }\label{sec-proof}
 
\emph{Notation}: $\mathscr{P}(\mathscr{X})$ denotes the set of all atomic PDs on $\mathscr{X}$ with finite number of atoms. The $n$-th Cartesian power is defined as the sample space of $\mathbf{X}=(X_1,\dots,X_n)$, with $P^{n}_{\mathbf{X}}$-PD determined in terms of the $n$-th Cartesian power of $P_X$. The cardinality of an alphabet is denoted by $\|\cdot\|$. For every $\delta>0$, we denote $\delta$-typical and conditional $\delta$-typical sets by $\mathscr{T}_{[X]_{\delta}}^{n}$ and $\mathscr{T}_{[Y|X]_{\delta}}^{n}(\mathbf{x})$, respectively.   \vspace{1mm}

\begin{proof} We first provide details of the proof of theorem \ref{theo-new-lowerbound1} where the general idea is as follows. We encode the message $m$ into a RV $U$, by using superposition and Marton coding we allow partial (or indirect \cite{ElGamal09}) decoding of $U$ via two other RVs, namely $(V_1,V_2)$. Hence receiver $Y_1$ indirectly decodes $U$ via $V_1$ while  receiver $Y_2$ indirectly decodes $U$ via $V_2$.

\emph{Code Generation:} Let $T_0\geq R$ and $S_i \geq 0$ with $i=1,2$.  Fix a PD of the require form $P_{XUV_1V_2|S} = P_{X|UV_1V_2S} P_{UV_1V_2|S} $ satisfying $ (U,V_1,V_2) \minuso  (X,S) \minuso (Y_1,Y_2)$ form a Markov chain. Randomly and independently generate $\lfloor 2^{nT_0} \rfloor$ sequences $\mb{u}(t_0)$ form $\mathscr{T}_{[U]_{\delta}}^{n}$ indexed by $t_0\in \{1,\dots, \lfloor 2^{nT_0} \rfloor\}$. Randomly partition the $\lfloor 2^{nT_0} \rfloor$ sequences into $\lfloor 2^{nR} \rfloor$ equal size bins. For each 
 $\mb{u}(t_0)$, randomly and independent generate:  (i)  $\lfloor 2^{nT_1} \rfloor$ sequences $\mb{v}_1(t_0,t_1)$  
  indexed by $t_1\in \{1,\dots, \lfloor 2^{nT_1} \rfloor\}$, each distributed uniformly  over the set $\mathscr{T}_{[V_1|U]_{\delta}}^{n}\big(\mb{u}(t_0)\big)$, (ii) $\lfloor 2^{nT_2} \rfloor$ sequences $\mb{v}_2(t_0,t_2)$  indexed by $t_2\in \{1,\dots, \lfloor 2^{nT_2} \rfloor\}$, each distributed uniformly over the set $\mathscr{T}_{[V_2|U]_{\delta}}^{n}\big(\mb{u}(t_0)\big)$.

\emph{Encoding:}  To send a message $m\in \{1,\dots, \lfloor 2^{nR} \rfloor\}$, choose an index $t_0^*\in \{1,\dots, \lfloor 2^{nT_0} \rfloor\}$ from the bin $m$ such that $\mb{u}(t_0^*)$ and  $\mb{s}$ are jointly typical, and choose indices  $t_1^*\in \{1,\dots, \lfloor 2^{nT_1} \rfloor\}$ and $t_2^*\in \{1,\dots, \lfloor 2^{nT_2} \rfloor\}$ such that $\mb{v}_1(t_0^*,t_1^*)$ and  $\mb{s}$ are jointly typical, $\mb{v}_2(t_0^*,t_2^*)$ and  $\mb{s}$ are jointly typical and the pair  $\big(\mb{v}_1(t_0^*,t_1^*),\mb{v}_2(t_0^*,t_2^*)\big)$ is jointly typical with high probability. Then send the codeword  $\mb{x}$ distributed uniformly over the set $\mathscr{T}_{[X|UV_1V_2S]_{\delta}}^{n}\big(\mb{u}(t_0^*),\mb{v}_1(t_0^*,t_1^*),\mb{v}_2(t_0^*,t_2^*),\mb{s} \big)$. To ensure the success of this coding, we require that 
\begin{align}                                                    
T_0 -R &> S_0,\,\,\,\,\,\,\, T_1\geq S_1\,\,\,\,\, \textrm{and} \,\,\,\,\, T_2\geq S_2 \nonumber\\ 
S_0 & > I(U;S),\nonumber\\
S_1+S_2 & >  I(V_1;V_2|U)+I(V_1,V_2;S|U), \nonumber\\
S_1 & > I(V_1;S|U),\,\,\,\,\,\,\, S_2  > I(V_2;S|U). 
\label{inequalities-constraint1A}
\end{align} 

\emph{Decoding:} Receiver $Y_1$  finds $t_0$ and thus the message $m$ via indirect decoding of $\mb{u}(t_0)$  based on  $\mb{v}_1(t_0,t_1)$.  Hence receiver $Y_1$ declares that $t_0\in \{1,\dots, \lfloor 2^{nT_0} \rfloor\}$ is sent if it is the unique index such that  $\mb{v}_1(t_0,t_1)$ and  $\mb{y}_1$ are jointly typical   $(\mb{u}(t_0),\mb{v}_1(t_0,t_1))\in \mathscr{T}_{[UV_1]_{\delta}}^{n}$ for some $t_1\in \{1,\dots, \lfloor 2^{nT_1} \rfloor\}$. This can be achieved with small probability of error  provided 
\begin{equation}
 T_0+T_1 < I(U,V_1;Y_1).   \label{inequalities-constraint2A}
\end{equation}
Notice that here  receiver $Y_1$  cannot correctly decode  $\mb{v}_1(t_0,t_1)$. Similarly, receiver $Y_2$  finds $t_0$ and thus the message $m$ via indirect decoding of $\mb{u}(t_0)$  based on  $\mb{v}_2(t_0,t_2)$.  Hence receiver $Y_2$ declares that $t_0\in \{1,\dots, \lfloor 2^{nT_0} \rfloor\}$ is sent if it is the unique index such that  $\mb{v}_2(t_0,t_2)$ and  $\mb{y}_2$ are jointly typical   $(\mb{u}(t_0),\mb{v}_2(t_0,t_2))\in \mathscr{T}_{[UV_2]_{\delta}}^{n}$ for some $t_2\in \{1,\dots, \lfloor 2^{nT_2} \rfloor\}$. This can be achieved with small probability of error  provided 
\begin{equation}
 T_0+T_2 < I(U,V_2;Y_2). \label{inequalities-constraint3A}
\end{equation}
Observe that receiver $Y_2$  cannot correctly decode  $\mb{v}_2(t_0,t_2)$. By applying the Fourier-Motzkin procedure to eliminate $(T_i,$ $S_i)_{\{i=0,1,2\}}$ from \eqref{inequalities-constraint1A}-\eqref{inequalities-constraint3A}, we obtain the following inequalities:
\begin{align} 
R &\leq I(U,V_1;Y_1)-I(U,V_1;S),\nonumber\\ 
R &\leq I(U,V_2;Y_2)-I(U,V_2;S),\nonumber\\
2R & \leq  I(U,V_1;Y_1)+I(U,V_2;Y_2)  - 2I(U;S) \nonumber\\
&-I(V_1;V_2|U)- I(V_1,V_2;S|U). 
\label{inequalities-proof}
\end{align} 
 This concludes the proof of the rate \eqref{capacity-new-lowerbound1}. We now provide details on the proof of the extended rate \eqref{capacity-new-lowerbound2}. The code generation, encoding and decoding remain very similar to the previous. Encoding succeeds with high probability as long as 
\begin{align}   
T_0 -R &> S_0,\,\,\,\,\,\,\, T_k\geq S_k\,\,\,\,\, \textrm{for all $k=\{1,\dots,K\}$}\nonumber\\ 
 \sum\limits_{k \in  \mc{K} } S_k &>  \sum\limits_{k \in  \mc{K}  } H(V_k | U ) - H\big(\{ V_t\, | \, t \in \mc{K}   \} | U,S\big), 
\label{inequalities-constraint2}
\end{align} 
for every subset $\mc{K} \subseteq \{1,\dots,K\}$ and $S_k\geq 0$ for all $k=\{1,\dots,K\}$. Decoding  succeeds with high probability if  
\begin{equation}
 T_0+T_k < I(U,V_k;Y_k),\,\,\,\,\,\,\,\,\,\,k=\{1,\dots,K\}.  \label{inequalities-constraint3}
\end{equation}
By combinning expressions \eqref{inequalities-constraint2} and  \eqref{inequalities-constraint3}, applying  Fourier-Motzkin procedure, it is not difficult to show the rate \eqref{capacity-new-lowerbound2}. 
\end{proof}

\section{Application Example:  Compound Gaussian Dirty-Paper channel} \label{sec-example}
In this section,  we begin by introducing the compound Gaussian Dirty-Paper (GDP) channel and then present lower and upper bounds on its capacity. Consider  a compound memoryless GDP channel whose output is given by  
 \begin{equation}
 \big\{\mb{Y}_{(\beta,\theta)} =\beta\!\cdot\!\mb{X} + \theta\!\cdot\!\mb{S} + \mb{Z}\big\}_{(\beta,\theta)\in\Theta}, \label{eq-multiplicative-channel}
 \end{equation}
where $\mb{X}=(X_{1},\dots,X_{n})$ is the channel input and  $\mb{S}$ is a white Gaussian interference (known to the transmitter only) of power $Q$ and independent of the sequence $\mb{Z}$ of white Gaussian noise of power $N$. The inputs must satisfy a limited-power constraint $P$ (often $\ll Q$), which takes the form 
$\esp \left[ \sum_{i=1}^n X_i^2(m,\mb{S})\right]\leq n P$, where the expectation is taken over the ensemble of messages and the interference sequence. 

We  focus on the case $\beta=\beta_0=1$, where the transmitter is unaware of $\theta\in\Theta$, assumed to take values from a set of real numbers, namely $\Theta\triangleq\{\theta_1,\dots,\theta_{\|\Theta \|}\}$. The fading coefficient $\theta$ remains fixed throughout a transmission.  
\subsection{Lower and Upper Bounds on the Capacity}
\begin{lemma}[Lower bound]\label{lemma-lower-compound}
A lower bound on the capacity of the compound GDP channel \eqref{eq-multiplicative-channel}  is given by
\begin{equation}
C_{\Theta}(P)\geq\max_{ (P_C,P_\Delta ):P_C\geq0,P_\Delta \geq 0,P_C+P_\Delta \leq P} R_{-}^{\Theta}(P_C,P_\Delta),\label{capacity-multiplicative-channel1}
\end{equation}
where $\theta_{\textrm{min}}\triangleq \min \{ \theta:\,\theta\in \Theta\}$, $\theta_{\textrm{max}}\triangleq \max \{ \theta:\,\theta\in \Theta\}$ and
\begin{align}
&R_{-}^{\Theta}(P_C,P_\Delta)=\displaystyle{\frac{1}{2\|\Theta \|}\log \left(1+\frac{P_\Delta}{N}\right)}  \nonumber\\
&+ \left\{\begin{array}{ll}
\displaystyle{\frac{1}{2}\log\left( 1+\frac{P_C}{P_\Delta+N+ \theta_{\textrm{min}}^2 Q}\right)} & \textrm{if $| \theta_{\textrm{min}}| = | \theta_{\textrm{max}}|$}\\
\displaystyle{\frac{1}{2}\log\left[ 1+\frac{P_C(1-\epsilon_\Theta)}{P_\Delta+N + \epsilon_\Theta P_C}\right]} & \textrm{if $| \theta_{\textrm{min}}| \neq |\theta_{\textrm{max}}|$},
\end{array}\right.     
\label{capacity-multiplicative-channel2}     
\end{align}       
 and the \emph{mismatch factor} $0\leq \epsilon_\Theta\leq 1$ is defined as          
\begin{equation*}
\epsilon_\Theta\triangleq \frac{ 1}{(\theta_{\textrm{min}}+\theta_{\textrm{max}} )^2} \left[ \sqrt{\theta_{\textrm{max}}^2+ \frac{T}{Q}}  -  \sqrt{\theta_{\textrm{min}}^2+ \frac{T}{Q}}\vspace{-2mm} \right]^2\label{capacity-multiplicative-channel3}
\end{equation*}
with $T=P_C+P_\Delta+N$. Optimizing expression \eqref{capacity-multiplicative-channel2} over $P_C$ and $P_\Delta$ subject to $P_C+P_\Delta \leq P$   yields the rate: 

\begin{equation}
R_{-}^{\Theta}(P)=\left\{ \begin{array}{l} 
\displaystyle{\frac{1}{2}\log\left[ 1+\frac{P(1-\epsilon_\Theta^\ast)}{N + \epsilon_\Theta^\ast P}\right]},\,\,\,\,  \textrm{if $\,\,\,\,  \epsilon_\Theta^\ast< \displaystyle{\frac{N(\|\Theta \|-1)}{P+\|\Theta \|N}}$} \\
\displaystyle{\frac{1}{2\|\Theta \|}\log\left[\frac{(P+N)}{\|\Theta \| N (1-\epsilon_\Theta^\ast)}  \left( \frac{\|\Theta \|-1}{\|\Theta \| \epsilon_\Theta^\ast} \right)^{\|\Theta \|-1} \right]}, \\ \,\,\,\, \,\,\,\, \,\,\,\, \,\,\,\, \,\,\,\, \,\,\,\, \,\,\,\, \,\,\,\, \,\,\,\,  \textrm{if $\,\,\,\, \displaystyle{ \frac{ N(\|\Theta \|-1)}{P+\|\Theta \|N}\leq \epsilon_\Theta^\ast<  \frac{\|\Theta \|-1}{\|\Theta \|}}$} \\
\displaystyle{\frac{1}{2\|\Theta \|}\log\left( 1+\frac{P}{N}\right)}, \,\,\,\, \textrm{if $\,\,\,\, \epsilon_\Theta^\ast\geq  \displaystyle{\frac{\|\Theta \|-1}{\|\Theta \|}}$} 
\end{array}\right.
\label{capacity-multiplicative-channel6B}
\end{equation}
where   

\begin{equation*}
\epsilon_\Theta^\ast\triangleq\frac{ 1}{(\theta_{\textrm{min}}+\theta_{\textrm{max}} )^2} \left[ \sqrt{\theta_{\textrm{max}}^2+ \frac{P+N}{Q}}  -  \sqrt{\theta_{\textrm{min}}^2+ \frac{P+N}{Q}} \right]^2, \label{capacity-multiplicative-channel3BB}
\end{equation*}
for $|\theta_{\textrm{min}}| \neq |\theta_{\textrm{max}}|$. For the case $|\theta_{\textrm{min}}| = |\theta_{\textrm{max}}|$,  
\begin{equation}
R_{-}^{\Theta}(P)=\left\{ \begin{array}{l} 
\displaystyle{\frac{1}{2}\log\left( 1+\frac{P}{N +  \theta_{\textrm{min}}^2 Q }\right)}, \,\,\,\,\,\,\,\,  \,\, \textrm{if $ \displaystyle{\frac{\theta_{\textrm{min}}^2}{(\|\Theta \|-1)}}    < \displaystyle{\frac{N}{Q}}$} \\
\displaystyle{\frac{1}{2\|\Theta \|}\log\left[ \frac{ (P+N+\theta_{\textrm{min}}^2 Q)^{\|\Theta \|} }{  \|\Theta \| N } \left(\frac{\|\Theta \|-1}{\|\Theta \|  \theta_{\textrm{min}}^2 Q } \right)^{M-1} \right]}, \\  \,\,\,\,\,\,\,\,  \,\, \,\,\,\,\,\,\,\,  \,\,\,\,\,\,\,\,   \,\,\,\,\,\,\,\,  \,\, \,\,\,\,\,\,\,\,  \,\,\textrm{if $ \displaystyle{\frac{N}{Q}\leq \frac{\theta_{\textrm{min}}^2}{(\|\Theta \|-1)}  < \frac{(P+N)}{Q}}$} \\
\displaystyle{\frac{1}{2\|\Theta \|}\log\left( 1+\frac{P}{N}\right)}, \,\,\,\, \,\,\,\, \textrm{if $\displaystyle{\frac{\theta_{\textrm{min}}^2}{(\|\Theta \|-1)} }   \geq  \displaystyle{\frac{(P+N)}{Q}}$.} 
\end{array}\right.
\label{capacity-multiplicative-channel6C}
\end{equation}
\end{lemma}

\begin{lemma}[Upper bound]\label{lemma-upper-compound}
An upper bound on the capacity of the compound GDP channel is given by       
\begin{align}
C_{\Theta}(P)\leq  & R_{+}^{\Theta}(P)\triangleq \max\limits_{\rho \in [-1,1]} \min \Big\{  \displaystyle{\frac{1}{2}\log\left[ 1+\frac{P(1-\rho^2)}{N}\right]},\nonumber \\ 
 & \displaystyle{\frac{1}{4}\log\left[\frac{ P+N+ \theta_{\textrm{max}}^2 Q + 2 \theta_{\textrm{max}} \rho \sqrt{PQ}   }{ \sqrt{( \theta_{\textrm{max}} -\theta_{\textrm{min}}  )^2NQ} }\right]} \nonumber \\
+&    \displaystyle{\frac{1}{4}\log\left[\frac{ P+N+ \theta_{\textrm{min}}^2 Q + 2 \theta_{\textrm{min}} \rho \sqrt{PQ}   }{ \sqrt{( \theta_{\textrm{max}} -\theta_{\textrm{min}}  )^2NQ} }\right]}                  \Big\},      \label{upper-bound-exp}
\end{align}
with $\theta_{\textrm{min}}\triangleq \min \{ \theta:\,\theta\in \Theta\}$ and $\theta_{\textrm{max}}\triangleq \max \{ \theta:\,\theta\in \Theta\}$.
\end{lemma}
The proof of lemma \ref{lemma-lower-compound} is sketched below while the proof of lemma  \ref{lemma-upper-compound}  follows similar to \cite{Khisti-Lapidoth-Wornell2007}. Observe that the \emph{mismatch factor} introduces the capacity loss due to the uncertainty at the encoder on the value of $\theta$. Hence for scenarios where the mismatch factor is smaller, e.g. $\theta_{\textrm{min}} \approx \theta_{\textrm{max}}$ or $(P+N)\gg Q  $, expression \eqref{capacity-multiplicative-channel2} becomes closer to the capacity when the encoder and the decoder are both aware of the channel index $\theta$ controlling the communication. 

For $\|\Theta \|=2$, the lower bound \eqref{capacity-multiplicative-channel6B} provides significative gains compared to the previous bound \cite{Khisti-Lapidoth-Wornell2007}.  Although the bound \eqref{upper-bound-exp} is not tight in general,  notice that it is a sharper bound than those derived by previous results in \cite{Mitran-Devroye-Tarokh2006,Khisti-Lapidoth-Wornell2007,Pulkit-Anant2007} and it is tight for some special sets $\Theta$ as shown in Figure \ref{fig}. 

\emph{Coding strategy:} Notice that when $\epsilon_\Theta^\ast< N(\|\Theta \|-1)/(P+\|\Theta \| N)$ the best encoder strategy is implementing a DPC to mitigate the common part of the interfering signal and hence the remainder part is treated as additional noise. In contrast to this, if $\epsilon_\Theta^\ast\geq  (\|\Theta \|-1)/\|\Theta \|$ the best encoder strategy becomes to use time-sharing to mitigate (completely) the interference. This is obtained by allowing the encoder and the decoder to have access to a source of common randomness (e.g. a dither sequence \cite{Erez-Shamai-Zamir2005}), which is not available if $X$ is restricted to be a deterministic mapping. Otherwise, when $N(\|\Theta \|-1)/(P+\|\Theta \|N) \leq \epsilon_\Theta^\ast<  (\|\Theta \|-1)/\|\Theta \|$   the encoder combines both strategies by using superposition coding.
\begin{figure}
\centering
\includegraphics[width=5in]{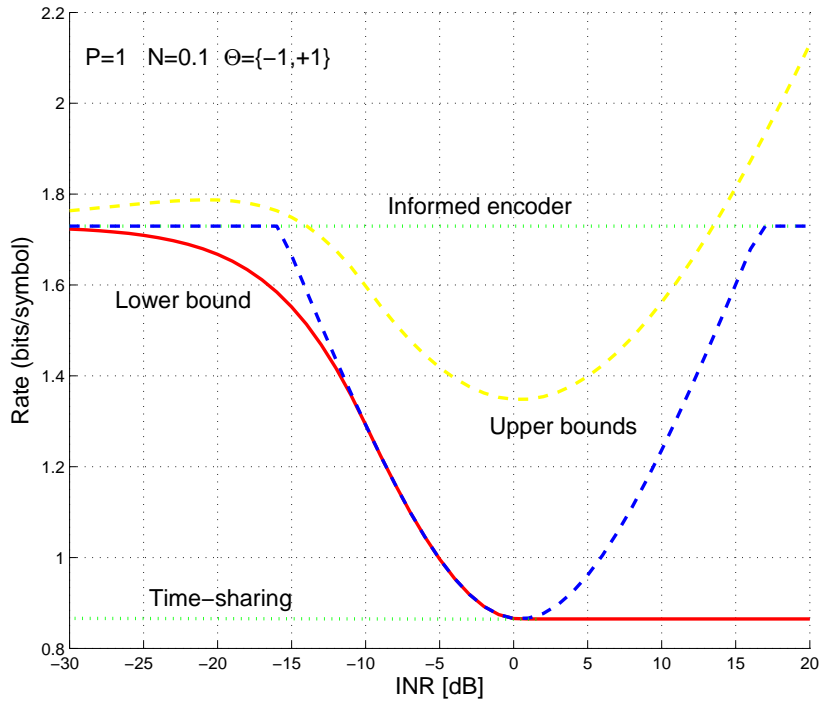}
\caption{Lower and upper bounds on the capacity of compound GDP channel $\Theta=\{-1,+1\}$, for $P=1$ and $N=0.1$, as a function of INR=$Q/N$. \vspace{-4mm}}
\label{fig}
\end{figure}    
\emph{Asymptotic analysis:} In the limit of high SNR (fixed $N$ and $Q$, $P_C+P_\Delta \to \infty $) the mismatch factor $\epsilon_\Theta$ vanishes to zero and thus the rate expression \eqref{capacity-multiplicative-channel2} coincides with its natural upper bound given by the case when the encoder is informed with $\theta$, which establishes the optimality of the lower bound in the high SNR limit.  In the limit when $Q\to \infty$ (for $N,P,\|\Theta \|$ fixed)  the mismatch factor becomes 
$$
 \epsilon_{\infty}^\ast \triangleq \lim\limits_{Q \to \infty} \epsilon_\Theta^\ast =\left( \frac{ \theta_{\textrm{max}}-\theta_{\textrm{min}} }{\theta_{\textrm{max}}+\theta_{\textrm{min}} }\right)^2. 
$$
Furthermore, when $\|\Theta \|\gg 1$ (for $N,P,Q$ fixed) the lower bound in \eqref{capacity-multiplicative-channel6B} reduces to 
 \begin{equation}
\lim\limits_{\|\Theta \| \to \infty} R_{-}^{\Theta}(P)=\frac{1}{2}\log\left[ 1+\frac{P(1- \epsilon_\Theta^\ast)}{N + \epsilon_\Theta^\ast  P}\right],            \label{eq-scenario1}
\end{equation}
while for $\theta_{\textrm{max}} \gg \theta_{\textrm{min}}$  (for $N,P,Q,K$ fixed) \eqref{capacity-multiplicative-channel6B} writes
 \begin{equation}
\lim\limits_{\theta_{\textrm{max}} /\theta_{\textrm{min}}  \to \infty} R_{-}^{\Theta}(P)= \frac{1}{2\|\Theta \|}\log\left( 1+\frac{P}{N}\right). \label{eq-scenario2}
\end{equation}
The scenarios \eqref{eq-scenario1} and \eqref{eq-scenario2}, i.e. $ \epsilon_{\infty}^\ast \approx 1$ and $\|\Theta \|\gg 1$, yield the most important loss of degrees of freedom.

\subsection{Sketch of Proof of theorem \ref{lemma-lower-compound}} 


\emph{Coding scheme:} The encoder splits the information $m=(m_0,m_1)$, namely common information $m_0$ and private information $m_1$. Then it divides the power $P$ into $\{P_C,P_1,\dots,P_{\|\Theta \|}\}$. The encoder sends $m_0$ using a standard DPC $\mb{U}$, sampled of length $n$ i.i.d. from a PD $\textrm{P}_{U|S}=\mc{N}(\alpha_c S, P_C)$, applied to the interference $\mb{S}$ and treats the reminder interference as noise. Whereas $m_1$ is sent using time-sharing via $\|\Theta \|$ different DPCs $\{\mb{V}_{1},\dots,\mb{V}_{K}\}$,  sampled i.i.d. of lengths $\{\lfloor n  \lambda_1\rfloor ,\dots,\lfloor n \lambda_K\rfloor\}$ from PDs $\textrm{P}_{V_{k}  |U S}=\mc{N}\big(\alpha_k(\theta_k-\alpha_c) S+U, P_k\big)$, applied once to each of interferences $\{\theta_1\mb{S},\dots, \theta_K\mb{S} \}$.  Send $\mb{X}=\mb{X}_C+\mb{X}_D$ with  $\mb{X}_C=\mb{U} - \alpha_c \mb{S}$ and $\mb{X}_D=[\mb{X}_{1} \dots \mb{X}_{K}]$,     where $\mb{X}_{k}=\mb{V}_{k}- \alpha_k(\theta_k-\alpha_c) \mb{S}-\mb{U}$. By substituting this in \eqref{capacity-new-lowerbound2}, it is not difficult to show that 
 \begin{align*} 
 &R_{-}^{\Theta}(P) = \max\limits_{(\alpha_c,\underline{\alpha})\in\mathbb{R}^{\|\Theta \|+1}}  \min\limits_{k\in\{1,\dots,\|\Theta \|\}} \Big\{ I(U^{(\alpha_c)};Y_k)- \nonumber\\ 
 &I(U^{(\alpha_c)};S) + \lambda_k \big[ I\big(V_k^{(\alpha_k)};Y_k|U^{(\alpha_c)}\big)-I\big(V_k^{(\alpha_k)};S|U^{(\alpha_c)}\big) \big]\Big\}.
 \label{rate-exp1}
 \end{align*} 

%


\section{Summary and Discussion}
We have investigated the compound state-dependent DMC with non-causal state information at the transmitter but not at the receiver. Some references \cite{Piantanida-shamai2009}, \cite{Moulin-Wang07}  and conjectures on the capacity of these channels  \cite{Mitran-Devroye-Tarokh2006} have lent support to the general belief that the natural extension \eqref{capacity-low} of the Gel'fand and Pinsker's capacity \cite{GF-1980}  to the compound setting case is indeed optimal. This paper shows that this is not in general the case. We found that the capacity of the general compound DMC can be strictly larger than the straightforward extension of the Gel'fand and Pinsker's capacity. We derived a new lower bound on the capacity  and showed that it is tight for the  compound channel with degraded components. It would be of interest to determine whether the result here \eqref{capacity-new-lowerbound1} is strictly better than the common rate result reported in \cite[eq. (45)]{Steinberg-Shamai-2005} and further explore the optimality of this result for channels with semi-deterministic and less noisy components.


The compound Gaussian Dirty-Paper channel that consists of an AWGN channel with an additive interference, where the input and the state signals are affected by fading coefficients whose realizations are unknown at the transmitter, was also considered. We derived lower and upper bounds on the capacity of this channel that are tight for some special cases.

\section*{Acknowledgment}
This research is supported by the FP7 Network of Excellence in Wireless COMmunications NEWCOM++.

\bibliographystyle{IEEEtran.bst}
\bibliography{biblio.bib}

\end{document}